%% file: ContextFreeConstrainedPathFindingInGraph.tex
\documentclass{sig-alternate} % V1.2

\usepackage{graphicx}
\usepackage{caption}
\usepackage{subcaption}
\usepackage{tikz}
\usepackage{mathtools}

\usepackage{graphicx}
\usepackage{hyperref}
\usepackage{textcomp}

\usepackage{algpseudocode}
\usepackage{algorithm}
\usepackage{algorithmicx}

\begin{document}

\newtheorem{mytheorem}{Theorem}
\newtheorem{lemma}{Lemma}

\algrenewcommand\algorithmicindent{0.5em}
\algnewcommand\algorithmicswitch{\textbf{switch}}
\algnewcommand\algorithmiccase{\textbf{case}}
\algnewcommand\algorithmicassert{\texttt{assert}}
\algnewcommand\Assert[1]{\State \algorithmicassert(#1)}
% New "environments"
\algdef{SE}[SWITCH]{Switch}{EndSwitch}[1]{\algorithmicswitch\ #1\ \algorithmicdo}{\algorithmicend\ \algorithmicswitch}
\algdef{SE}[CASE]{Case}{EndCase}[1]{\algorithmiccase\ #1}{\algorithmicend\ \algorithmiccase}

\algtext*{EndSwitch}
\algtext*{EndCase}
\algtext*{EndWhile}% Remove "end while" text
\algtext*{EndIf}% Remove "end if" text
\algtext*{EndFor}% Remove "end for" text
\algtext*{EndFunction}% Remove "end function" text

\newif\ifboldnumber
\newcommand{\boldnext}{\global\boldnumbertrue}

% Default definition is \footnotesize#1:
\algrenewcommand\alglinenumber[1]{%
  \footnotesize\ifboldnumber\bfseries\fi\global\boldnumberfalse#1:}

\makeatletter
\def\@copyrightspace{\relax}
\makeatother

%\title{Generalized LL Parsing for Context-Free Constrained Path Search Problem}
%\title{Structural representation of result for context-free path query problem}
%\title{Generalized LL for context-free path query problem}
%\title{Context-free path query with Structural representation of result}
%\title{Generalized LL for Context-free path query with Structural representation of result}
\title{Context-Free Path Querying with Structural Representation of Result}

\sloppy

\numberofauthors{2}

\author{
\alignauthor
       Semyon Grigorev\\
       \affaddr{Saint Petersburg State University}\\
       \affaddr{7/9 Universitetskaya nab.}\\
       \affaddr{St. Petersburg, 199034 Russia}\\
       \email{semen.grigorev@jetbrains.com}
\alignauthor
       Anastasiya Ragozina\\
       \affaddr{Saint Petersburg State University}\\
       \affaddr{7/9 Universitetskaya nab.}\\
       \affaddr{St. Petersburg, 199034 Russia}\\
       \email{ragozina.anastasiya@gmail.com}
}

\maketitle

\begin{abstract}
Graph data model and graph databases are very popular in various areas such as bioinformatics, semantic web, and social networks.
One specific problem in the area is a path querying with constraints formulated in terms of formal grammars.
The query in this approach is written as grammar, and paths querying is graph parsing with respect to given grammar.
There are several solutions to it, but how to provide structural representation of query result which is practical for answer processing and debugging is still an open problem.
In this paper we propose a graph parsing technique which allows one to build such representation with respect to given grammar in polynomial time and space for arbitrary context-free grammar and graph.
Proposed algorithm is based on generalized LL parsing algorithm, while previous solutions are based mostly on CYK or Earley algorithms, which reduces  time complexity in some cases

\end{abstract}

%\category{Information Systems}{Database Management}{Query Languages}
%A category including the fourth, optional field follows...
%\category{Theory of computation}{Formal languages and automata theory}{Grammars and context-free languages}

\terms{Languages, Algorithms}


\section{Introduction}
Graph data model and graph data bases are very popular in various areas such as bioinformatics, semantic web, social networks, etc.
Extraction of paths which satisfy specific constraints may be useful for investigation of graph structured data and for detection of relations between data items.
One specific problem---path querying with constraints---is usually formulated in terms of formal grammars and is called formal language constrained path problem~\cite{FLCpathProblem}.

%!!!!!!!!!!!!!!!!!!!!!!!!!!!!!!!!!!!!!!!!!!!!
%Information from different areas such as bioinformatic, semantic web, social networks can be represented in graph model. 
%Moreover, there are data graph data bases. (?) One of the common graph problem is paths extraction from graph. 
%Paths must satisfy specific constrains and the search must use a reasonable time. (?)     
%!!!!!!!!!!!!!!!!!!!!!!!!!!!!!!!!!!!!!!!!!!!!

Classical parsing techniques can be used to solve formal language constrained path problem.
It means that such technique can be used for more common problem---graph parsing. 
Graph parsing may be required in graph data base querying, formal verification, string-embedded language processing, and another areas where graph structured data is used. 

Existing solutions in databases field usually employ such parsing algorithms as CYK or Earley(for example~\cite{ConjCFPathQuery}, ~\cite{GraphQueryWithEarley}). 
These algorithms have nonlinear time complexity for unambiguous grammars ($O(n^3)$ and $O(n^2)$ respectively).
Moreover, in case of CYK, the input grammar should be transformed to Chomsky normal form (CNF) which leads to grammar size increase.
To solve these problems, one can use such parsing algorithms as GLR and GLL which have cubic worst-case time complexity and linear complexity for unambiguous grammars.
Also there is no need to transform a grammar to CNF for these algorithms.
These facts allow us to improve performance of parsing in some cases.

Despite the fact that there is a set of path querying solutions~\cite{GraphQueryWithEarley, ConjCFPathQuery, QueryGraphWithData, RegularDBQuery}, query result exploration is still a challenge~\cite{hofman2015separabilityForRegQueryDebugging}, as also a simplification of complex query debugging.
Structural representation of query result can be used to solve these problems, and classical parsing techniques provide such representation---derivation tree---which contains exhaustive information about parsed sentence structure in terms of specified grammar.

Graph parsing can also be used to analyze dynamically generated strings or string-embedded languages. 
String variable in a program may gets multiple values in run time.
In order to convey statical analysis, value set of string variable can be over-approximated with regular language which is represented as a finite automaton.
Moreover, to check a syntactic correctness of dynamically generated strings, one should check that all generated strings (all paths from start states to final states in the given automaton) are correct with respect to the given context-free grammar. 
There are solutions to this problem: GLR-based checker of string-embedded SQL queries~\cite{Alvor1, Alvor2},
 parser of string-embedded languages~\cite{relaxedRNGLR} based on RNGLR parsing algorithm.
RNGLR-based algorithm allows to construct derivation forest (i.e. the set of derivation trees) for all correct paths in the input automaton.

In this paper we propose a graph parsing technique which allows one to construct structural representation of query result with respect to the given grammar.
This structure can be useful for query debugging and exploration. 
Proposed algorithm is based on generalized top-down parsing algorithm---GLL~\cite{scott2010gll}---which has cubic worst-case time complexity and linear time complexity for LL grammars on linear input.  

\input{Preliminaries.tex}
\input{MotivExample.tex}
\input{Gll.tex}
\input{Complexity.tex}
\input{Example.tex}
\input{Evaluation.tex}
\input{Conclusion.tex}

\bibliographystyle{abbrv}
\bibliography{ContextFreeConstrainedPathFindingInGraph}

\input{appendix.tex}

\balancecolumns

\end{document}

%% file: Preliminaries.tex
\section{Preliminaries}

In this work we are focused on the parsing algorithm, and not on the data representation, and we assume that whole input graph can be located in RAM memory in the optimal for our algorithm way.

We start by introduction of necessary definitions.
\begin{itemize}
  \item Context-free grammar is a quadruple $G=(N, \Sigma, P, S)$, where $N$ is a set of nonterminal symbols, $\Sigma$ is a set of terminal symbols, $S \in N$ is a start nonterminal, and $P$ is a set of productions. 
  \item $\mathcal{L}(G)$ denotes a language specified by grammar $G$, and is a set of terminal strings derived from start nonterminal of $G$: $L(G) = \{\omega | S \Rightarrow_{G}^{*} \omega\}$.
  \item Directed graph is a triple $M = (V,E,L)$, where $V$ is a set of vertices, $L \subseteq \Sigma$ is a set of labels, and a set of edges $E\subseteq V\times L\times V$. 
  We assume that there are no parallel edges with equal labels: for every $e_1=(v_1,l_1,v_2) \in E, e_2=(u_1,l_2,u_2) \in E$ if $v_1 = u_1$ and $v_2 = u_2$ then $l_1 \neq l_2$.
  \item $tag: E \rightarrow L$ is a helper function which allows to get tag of edge. $$tag(e = (v_1,l,v_2), e \in E) = l$$
  \item $\oplus: L^+ \times L^+ \rightarrow L^+$ denotes a tag concatenation operation.
  \item Path $p$ in graph $M$ is a list of incident edges: 
  \begin{align*}
   p &= e_0,e_1,\dots,e_{n-1} \\
     &= (v_0,l_0,v_1),(v_1,l_1,v_2),\dots,(v_{n-1},l_{n-1},v_n)
  \end{align*}
  where $v_i \in V$, $e_i \in E$, $e_i=(v_i,l_i,v_{i+1})$, $l_i \in L$, $|p| = n, n \geq 1$. 
  \item $P$  is a set of paths $\{p: p \text{ path in } M\}$, where $M$ is a directed graph.
  \item $\Omega: P \rightarrow L^+$ is a helper function which constructs a string produced by the given path. For every $p \in P$
  \begin{align*}
  & \Omega(p = e_{0},e_{1},\dots,e_{n-1}) = \\
  & tag (e_{0}) \oplus \dots \oplus tag (e_{n-1}).
  \end{align*}
\end{itemize}

Using these definitions, we state the context-free language constrained path querying as, given a query in form of grammar $G$, to construct the set of paths $$P=\{p|\Omega(p) \in \mathcal{L}(G)\}.$$

Note that, in some cases, $P$ can be an infinite set, and hence it cannot be represented explicitly. 
In order to solve this problem, in this paper, we construct compact data structure representation which stores all elements of $P$ in finite amount of space and allows to extract any of them.

%% file: MotivExample.tex
\section{Motivating Example}\label{motivExample}

Suppose that you are student in a School of Magic.
It is your first day at School, so navigation in the building is a problem for you.
Fortunately, you have a map of the building (fig.~\ref{input}) and additional knowledge about building construction:
\begin{itemize}
  \item there are towers in the school (depicted as nodes of the graph in your map);
  \item towers can be connected by one-way galleries (represented as edges in your map);
  \item galleries have a ``magic'' property: you can start from any floor, but by following each gallery you either end up one floor above (edge label is `a'), or one floor below (edge label is `b'). 
\end{itemize}

\begin{figure}[h]
    \begin{center}
        \includegraphics[width=6cm]{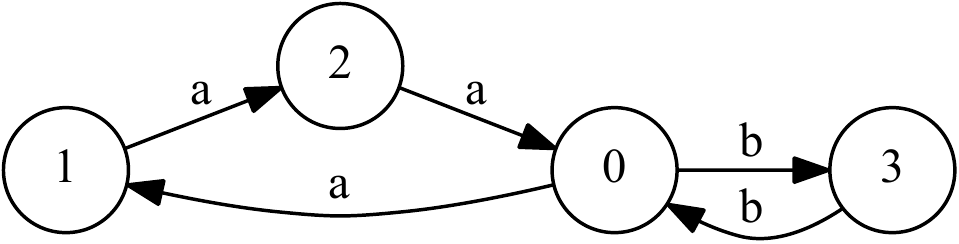}
        \caption{The map of School (input graph $M$)}
        \label{input}        
    \end{center}
\end{figure}

You want to find a path from your current position to the same floor in another tower. 
Map with all such paths can help you.
But orienteering is not your forte, so it would be great if the structure of the paths were as simple as possible and all paths had additional checkpoints to control your rout.

It is evident that the simplest structure of required paths is $\{ab, aabb, aaabbb, \dots\}$.
In terms of our definitions, it is necessary to find all paths $p$ such that $\Omega(p) \in \{a^n b^n, n \geq 1\}$ in the graph $M=(\{0;1;2;3\},E,\{a;b\})$ (figure~\ref{input}).

Unfortunately, language $\mathcal{L} = \{a^n b^n; n \geq 1\}$ is not regular which restricts the set of tools you can use. 
Another problem is the infinite size of solution, but, being incapable to comprehend an infinite set of paths, you want to get a finite map.  
Moreover, you want to know structure of paths in terms of checkpoints.

We are not aware of any existing tools which can solve this problem, thus we have created such tool.
Let us show how to get a map which helps to navigate in this strange School.

Fortunately, the language $\mathcal{L} = \{a^n b^n; n \geq 1\}$ is a context-free language and it can be specified with context-free grammar. 
The fact that one language can be described with multiple grammars allows to add checkpoints: additional nonterminals can mark required parts of sentences.
In our case, desired checkpoint can be in the middle of the path.
As a result, required language can be specified by the grammar $G_1$ presented in figure~\ref{grammarG}, where $N = \{s; \text{\textit{Middle}}\}$, $\Sigma = \{a; b\}$, and $S$ is a start nonterminal.

\begin{figure}[h]
   \begin{center}
   \[
\begin{array}{rl}
   0:& S \rightarrow a \ S \ b \\
   1:& S \rightarrow Middle \\
   2:& Middle \rightarrow a \ b
\end{array}
\]

   \caption{Grammar $G_1$ for language $L=\{a^n b^n; n \geq 1\}$ with additional marker for the middle of a path}
   \label{grammarG}        
   \end{center}
\end{figure}

In the next section, we present a graph parsing algorithm  which can be applicable to this kind of problems.

%% file: Gll.tex
\section{Graph Parsing Algorithm}

We propose a graph parsing algorithm which allows to construct finite representation of parse forest which contains derivation trees for all matched paths in graph.
Finite representation of result set with respect to the specified grammar may be useful not only for results understanding and processing, but also for query debugging. 

Our solution is based on generalized LL (GLL)~\cite{scott2010gll, FastPracticalGLL} parsing algorithm which allows to process arbitrary (including left-recursive and ambiguous) context-free grammars with worst-case cubic time complexity and linear time complexity for LL grammars on a linear input. 

\subsection{Generalized LL Parsing Algorithm}

Classical LL algorithm operates with a pointer to input (position $i$) and with a grammar slot---pointer to grammar in form $N \rightarrow \alpha \cdot x \beta $.
Parsing may be described as a transition of these pointers from the initial position ($i = 0$, $S \rightarrow \cdot \beta $, where $S$ is start nonterminal) to the final ($i = input.Length$, $s \rightarrow \beta \cdot$).
At every step, there are four possible cases in processing of these pointers. 

\begin{enumerate}
\item $N \rightarrow \alpha \cdot x \beta $, when $x$ is a terminal and $x = input[i]$. In this case both pointers should be moved to the right ($i \leftarrow i + 1$, $N \rightarrow \alpha  x \cdot \beta $).
\item $N \rightarrow \alpha \cdot X \beta $, when $X$ is nonterminal. In this case we push return address $N \rightarrow \alpha X \cdot \beta $ to stack and move pointer in grammar to position $X \rightarrow \cdot \gamma$.\label{itm:2}
\item $N \rightarrow \alpha \cdot $. This case means that processing of nonterminal $N$ is finished. We should pop return address from stack and use it as new slot.\label{itm:3}
\item $S \rightarrow \alpha \cdot $, where $S$ is a start nonterminal of grammar. In this case we should report success if $i = input.Length - 1$ or failure otherwise. 
\end{enumerate}

In the second case there can be several slots $X \rightarrow \cdot \gamma$, so a strategy on how to choose one of them to continue parsing is needed.
In LL(k) algorithm lookahead is used, but this strategy is still not good enough because there are context-free languages for which deterministic choice is impossible even for infinite lookahead~\cite{LLnonLL}.
On the contrary to LL(k), generalized LL does not choose at all, handling all possible variants.
Note, that instead of immediate processing of all variants, GLL uses descriptors mechanism to store all possible branches and process them sequentially. 
Descriptor is a quadruple $(L, s, j, a)$ where $L$ is a grammar slot, $s$ is a stack node, $j$ is a position in the input string, and $a$ is a node of derivation tree. 

The stack in parsing process is used to store return information for the parser---a name of function which will be called when current function finishes computation. 
As mentioned before, generalized parsers process all possible derivation branches and parser must store it's own stack for every branch. 
It leads to an infinite stack growth being done naively.  
Tomita-style graph structured stack (GSS)~\cite{Tomita} combines stacks resolving this problem.
Each GSS node contains a pair of position in input and a grammar slot in GLL . 

In order to provide termination and correctness, we should avoid duplication of descriptors, and be able to process GSS nodes in arbitrary order. It is necessary to use the following additional sets for this.
\begin{itemize}
\item $R$---working set which contains descriptors to be processed. Algorithm terminates whenever $R$ is empty.
\item $U$---all created descriptors. Each time when we want to add a new descriptor to $R$, we try to find it in this set first.
This way we process each descriptor only once which guarantee termination of parsing.
\item $P$---popped nodes. Allows to process descriptors (and GSS nodes) in arbitrary order. 
\end{itemize}

Instead of explicit code generation used in classical algorithm, we use table version of GLL~\cite{TableGLL} in order to simplify adaptation to graph processing.
As a result, main control function is different from the original one because it should process LL-like table instead of switching between generated parsing functions.
Control functions of the table based GLL are presented in Algorithm~\ref{mainTblFunctions}.
All other functions are the same as in the original algorithm and their descriptions can be found in the original article~\cite{scott2010gll} or in Appendix~\ref{GLLCode}.

\begin{algorithm}[h]
\begin{algorithmic}[1]
\caption{Control functions of table version of GLL}
\label{mainTblFunctions}
\Function{dispatcher}{ \ }
  \If{$R.Count \neq 0$}  
      \State{$(L,v,i,cN) \gets R.Get()$}
      \State{$cR \gets dummy$}
      \State{$dispatch \gets false$}
  \Else
      \State{$stop \gets true$}
  \EndIf
\EndFunction

\Function{processing}{ \ }
  \State{$dispatch \gets true$}
  \Switch{$L$}
  \Case{$(X \rightarrow \alpha \cdot x \beta)$ where $x = input[i + 1])$}
       \If{$cN = dummyAST$} 
          \State{$cN \gets \Call{getNodeT}{i}$} 
       \Else 
          \State{$cR \gets \Call{getNodeT}{i}$}
       \EndIf
       \State{$i \gets i + 1$}
       \State{$L \gets (X \rightarrow \alpha x \cdot \beta)$}
       \If{$cR \neq dummy$}
          \State{$cN \gets \Call{getNodeP}{L, cN, cR}$} 
       \EndIf
       \State{$dispatch \gets false$}        
  \EndCase
  \Case{$(X \rightarrow \alpha \cdot x \beta)$ where $x$ is nonterminal}
       \State{$v \gets$ \Call{create}{$(X \rightarrow \alpha x \cdot \beta), v, i, cN$}}
       \State{$slots \gets pTable[x][input[i]]$}
       \ForAll{$L \in slots$}
          \State{\Call{add}{$L,v,i,dummy$}} 
       \EndFor
  \EndCase
  \Case{$(X \rightarrow \alpha \cdot )$}
       \State{\Call{pop}{v,i,cN}} 
  \EndCase
  \Case{$(S \rightarrow \alpha \cdot )$ when $S$ is start nonterminal}
       \State{final result processing and error notification} 
  \EndCase
  \EndSwitch
\EndFunction

\Function{control}{}
  \While{not $stop$}  
      \If{$dispatch$}
        \State{\Call{dispatcher}{ \ }}
      \Else
         \State{\Call{processing}{ \ }}
      \EndIf
  \EndWhile
\EndFunction

\end{algorithmic}
\end{algorithm}

There can be more than one derivation tree of a string with relation to ambiguous grammar.
Generalized LL build all such trees and compact them in a special data structure Shared Packed Parse Forest~\cite{SPPF}, which will be described in the following section.

\subsection{Shared Packed Parse Forest}

Binarized Shared Packed Parse Forest (SPPF)~\cite{brnglr} compresses derivation trees optimally reusing common nodes and subtrees.
Version of GLL which uses this structure for parsing forest representation achieves worst-case cubic space complexity~\cite{gllParsingTree}.

Let us present an example of SPPF for the input sentence \verb|"ababab"| and ambiguous grammar $G_0$ (fig~\ref{grammarG0}).

\begin{figure}[h]
   \begin{center}
   \[
\begin{array}{rl}

   0: & S \rightarrow \varepsilon  \\
   1: & S \rightarrow a \ S \ b \\
   2: & S \rightarrow S \ S    
\end{array}
\]
   \caption{Grammar $G_0$}
   \label{grammarG0}        
   \end{center}
\end{figure}

There are two different leftmost derivations of the given sentence w.r.t. grammar $G_0$, hence SPPF contains two different derivation trees. Resulting SPPF(fig.~\ref{sppf}) and two trees extracted from it (fig.~\ref{tree1} and fig.~\ref{tree2}) are presented in the figure~\ref{sppfSample}. 
 
\begin{figure*}[ht]
    \begin{center}
    \centering
    \begin{subfigure}[b]{0.3\textwidth}
        \includegraphics[width=\textwidth]{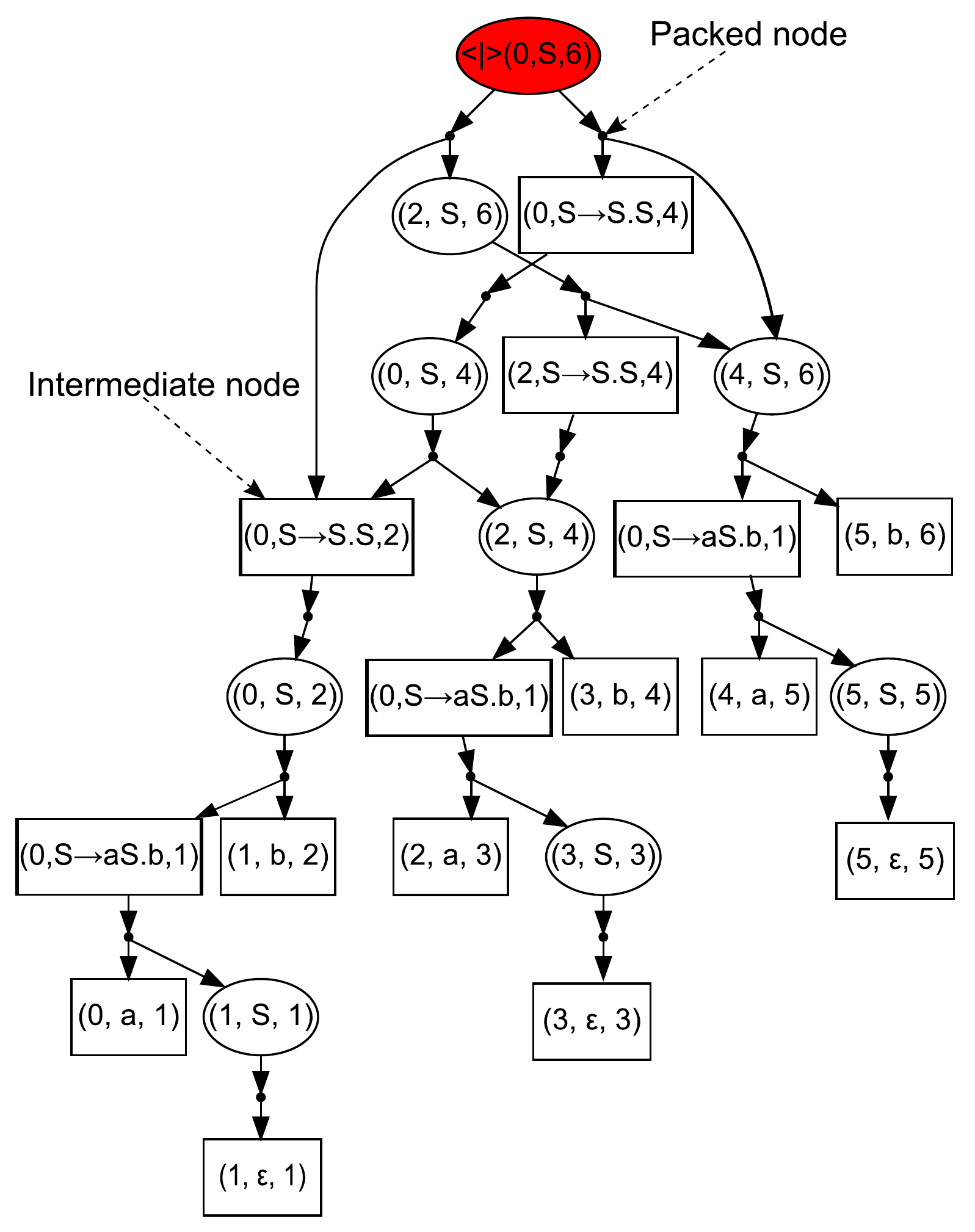}
        \caption{SPPF}
        \label{sppf}        
    \end{subfigure}
    ~
    \begin{subfigure}[b]{0.3\textwidth}
        \includegraphics[width=\textwidth]{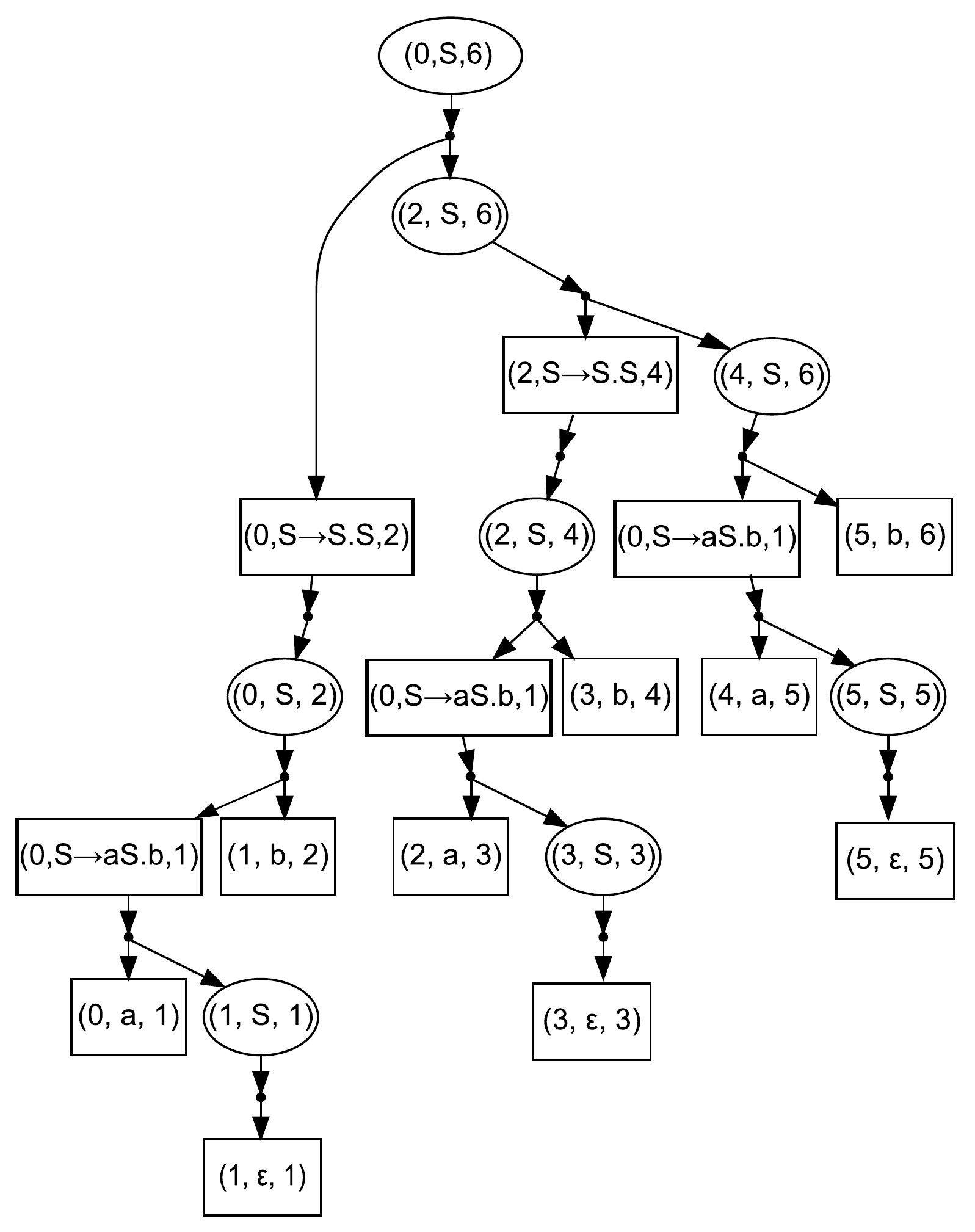}
        \caption{First derivation tree}
        \label{tree1}        
    \end{subfigure}
    ~
    \begin{subfigure}[b]{0.3\textwidth}
        \includegraphics[width=\textwidth]{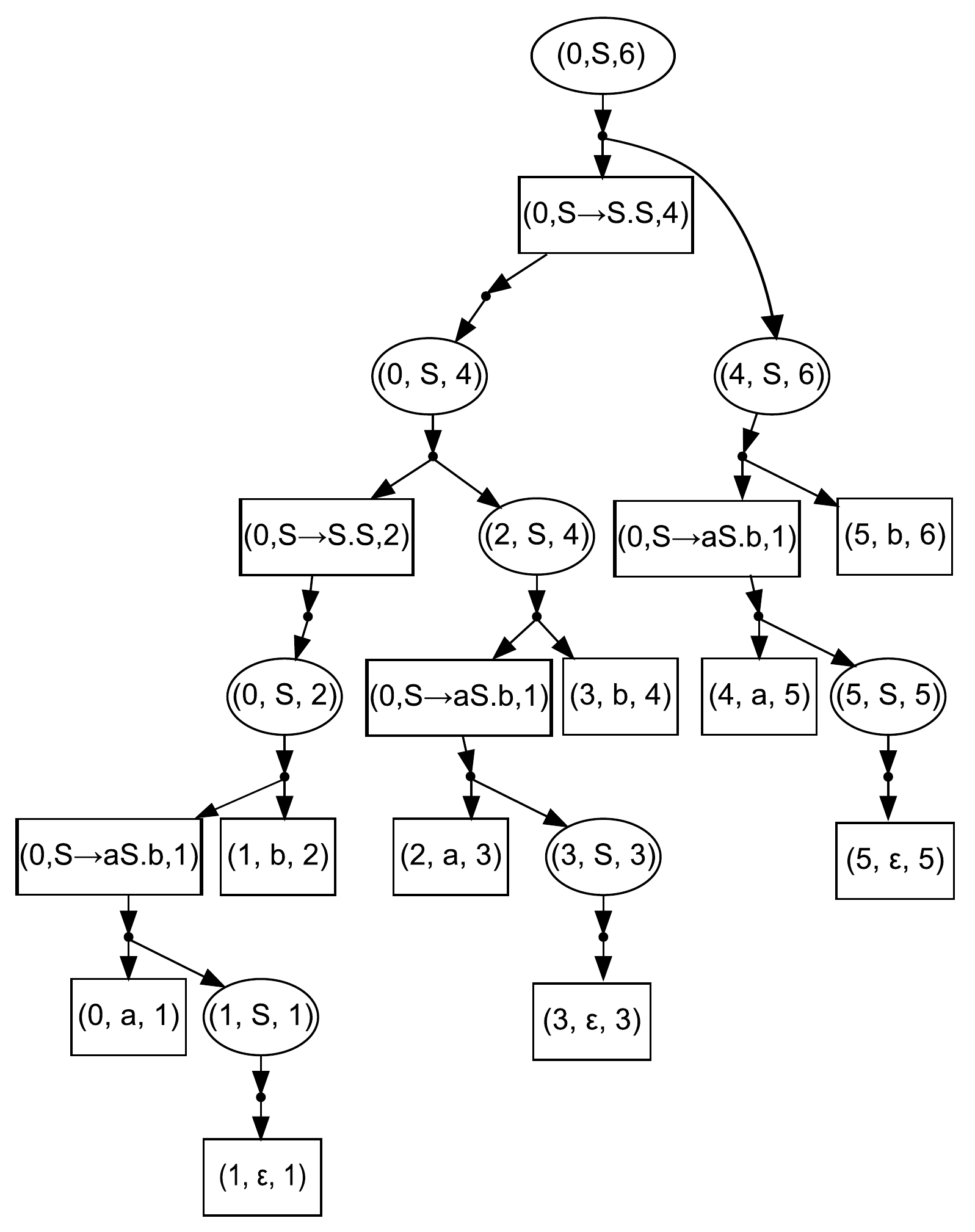}
        \caption{Second derivation tree}
        \label{tree2}        
    \end{subfigure}
    \caption{SPPF for sentence \textbf{\texttt{"ababab"}} and grammar $G_0$}
    \label{sppfSample}
    \end{center}                
\end{figure*}

Binarized SPPF can be represented as a graph in which each node has one of four types described below with correspondent graphical notation.
Let $i$ and $j$ be the start and the end positions of substring, and let us call a tuple $(i,j)$ an \textit{extension} of node.

\begin{itemize}
    \item Node of rectangle shape labeled with $(i, T, j)$ is a terminal node.     
    \item Node of oval shape labeled with $(i, N, j)$ is a nonterminal node. 
    This node denotes that there is at least one derivation for substring $\alpha=\omega[i..j-1]$ such that $N \Rightarrow^*_G \alpha, \alpha = \omega[i..j-1] $.
    All derivation trees for the given substring and nonterminal can be extracted from SPPF by left-to-right top-down graph traversal started from respective node. 
    We use filled shape and label of form $(<\mkern-11mu | \mkern-11mu> (i, N, j))$ to denote that there are multiple derivations from nonterminal $N$ for substring $\omega[i..j-1]$.
    \item Node of rectangle shape labeled with $(i,t,j)$, where $t$ is a grammar slot, is an intermediate node: a special kind of node used for binarization of SPPF.
    \item Packed node labeled with $(N \rightarrow \alpha \cdot \beta, k)$. In our pictures, we use dot shape for these nodes and omit labels because they are important only on SPPF constriction stage.
    Subgraph with ``root'' in such node is one variant of derivation from nonterminal $N$ in case when the parent is a nonterminal node labeled with $(<\mkern-9mu | \mkern-9mu> (i, N, j))$.

\end{itemize}

In our examples we remove redundant intermediate and packed nodes from the SPPF to simplify it and to decrease the size of structure.

\subsection{GLL-based Graph Parsing}

In this section we present such modification of GLL algorithm, that for input graph $M$, set of start vertices $V_s\subseteq V$, set of final vertices $V_f\subseteq V$, and grammar $G_1$, it returns SPPF which contains all derivation trees for all paths $p$ in $M$, such that $\Omega(p) \in L(G_1)$, and $p.start \in V_s,\ p.end \in V_f$.
In other words, we propose GLL-based algorithm which can solve language constrained path problem.

First of all, notice that an input string for classical parser can be represented as a linear graph, and positions in the input are vertices of this graph.
This observation can be generalized to arbitrary graph with remark that for a position there is a set of labels of all outgoing edges for given vertex instead of just one next symbol. 
Thus, in order to use GLL for graph parsing we need to use graph vertices as positions in input and modify \textbf{Processing} function to process multiple ``next symbols''.
Required modifications are presented in the Algorithm~\ref{modifAlgo}~(line \textbf{5} and \textbf{17}).
Small modification is also required for initialization of $R$ set: it is necessary to add not only one initial descriptor but the set of descriptors for all vertices in $V_s$.
All other functions are reused from original algorithm without any changes.

\begin{algorithm}[h]
\begin{algorithmic}[1]
\caption{\textbf{Processing} function modified in order to process arbitrary directed graph}
\label{modifAlgo}
\Function{processing}{\ }
  \State{$dispatch \gets true$}
  \Switch{$L$}
  \Case{$(X \rightarrow \alpha \cdot x \beta)$ where $x$ is terminal}
       \boldnext
       \ForAll{$\{ e | e \in input.outEdges(i), tag(e) = x \}$}
       \State{$new\_cN \gets cN$}
       \If{$new\_cN = dummyAST$} 
          \State{$new\_cN \gets \Call{getNodeT}{e}$} 
       \Else 
          \State{$new\_cR \gets \Call{getNodeT}{e}$}
       \EndIf
       \State{$L \gets (X \rightarrow \alpha x \cdot \beta)$}
       \If{$new\_cR \neq dummy$}
          \State{$new\_cN \gets \Call{getNodeP}{L, new\_cN, new\_cR}$} 
       \EndIf
       \State{\Call{add}{$L,v,target(e),new\_cN$}}
       \EndFor
  \EndCase
  \Case{$(X \rightarrow \alpha \cdot x \beta)$ where $x$ is nonterminal}
       \State{$v \gets$ \Call{create}{$(X \rightarrow \alpha x \cdot \beta), v, i, cN$}}
       \boldnext
       \State{$slots \gets \bigcup_{e \in input.OutEdges(i)} pTable[x][e.Token]$}
       \ForAll{$L \in slots$}
          \State{\Call{add}{$L,v,i,dummy$}} 
       \EndFor
  \EndCase
  \Case{$(X \rightarrow \alpha \cdot )$}
       \State{\Call{pop}{$v,i,cN$}} 
  \EndCase
  \Case{$\_$}
       \State{final result processing and error notification} 
  \EndCase
  \EndSwitch
\EndFunction

\end{algorithmic}
\end{algorithm}

Note that our solution handles arbitrary numbers of start and final vertices, which allows one to solve different kinds of problems arising in the field, namely all paths in graph, all paths from specified vertex, all paths between specified vertices. 
Also SPPF represents a structure of paths in terms of grammar which provides exhaustive information about result. 

Note that termination of proposed algorithm is inherited from the basic GLL algorithm.
We process finite graphs, hence the set of positions is finite, and tree construction has not been changed. 
As a result, the total number of descriptors is finite, and each of them is added in $R$ only once, thus main loop is finite.

%% file: Complexity.tex
\subsection{Complexity}

Time complexity estimation in terms of input graph and grammar size is quite similar to the estimation of GLL complexity provided in~\cite{gllParsingTree}.

\begin{lemma}\label{lem:Descriptors}
For any descriptor $(L,u,i,w)$ either $w = \$$ or $w$ has extension $(j,i)$ where u has index $j$.
\end{lemma}
\begin{proof}
Proof of this lemma is the same as provided for original GLL in~\cite{gllParsingTree} because main function used for descriptors creation has not been changed.
\end{proof}

\begin{mytheorem}\label{thm:GSSSpace}
The GSS generated by GLL-based graph parsing algorithm for grammar $G$ and input graph $M=(V,E,L)$ has at most $O(|V|)$ vertices and $O(|V|^2)$ edges.
\end{mytheorem}

\begin{proof}

Proof is the same as the proof of \textbf{Theorem 2} from~\cite{gllParsingTree} because structure of GSS has not been changed. 

\end{proof}

\begin{mytheorem}\label{thm:SPPFSpace}
The SPPF generated by GLL-based graph parsing algorithm on input graph $M=(V, E, L)$ has at most $O(|V|^3 + |E|)$ vertices and edges.
\end{mytheorem}

\begin{proof}
Let us estimate the number of nodes of each type.
\begin{itemize}
\item \textbf{Terminal nodes} are labeled with $(v_0, T, v_1)$, and such label can only be created if there is such $e \in E$ that $e=(v_0, T,v_1)$. 
Note, that there are no duplicate edges. 
Hence there are at most $|E|$ terminal nodes.
\item \textbf{$\varepsilon$-nodes} are labeled with $(v, \varepsilon, v)$, hence there are at most $|V|$ of them. 
\item \textbf{Nonterminal nodes} have labels of form $(v_0, N, v_1)$, so there are at most $O(|V|^2)$ of them.
\item \textbf{Intermediate nodes} have labels of form $(v_0, t, v_1)$, where $t$ is a grammar slot, so there are at most $O(|V|^2)$ of them.
\item \textbf{Packed nodes} are children either of intermediate or nonterminal nodes and have label of form $(N \rightarrow \alpha \cdot \beta, v)$.
There are at most $O(|V|^2)$ parents for packed nodes and each of them can have at most $O(|V|)$ children.
\end{itemize}

As a result, there are at most $O(|V|^3 + |E|)$ nodes in SPPF.

The packed nodes have at most two children so there are at most $O(|V|^3 + |E|)$ edges which source is packed node. 
Nonterminal and intermediate nodes have at most $O(|V|)$ children and all of them are packed nodes.
Thus there are at most $O(|V|^3)$ edges with source in nonterminal or intermediate nodes. As a result there are at most $O(|V|^3 + |E|)$ edges in SPPF.

\end{proof}

\begin{mytheorem}
The worst-case space complexity of GLL-based graph parsing algorithm for graph $M=(V,E,L)$ is $O(|V|^3 + |E|)$.
\end{mytheorem}

%\begin{proof}

Immediately follows from theorems~\ref{thm:GSSSpace} and~\ref{thm:SPPFSpace}. 

%\end{proof}

\begin{mytheorem}\label{thm:complexity}
The worst-case runtime complexity of GLL-based graph parsing algorithm for graph $M=(V,E,L)$ is $$O\left(|V|^3*\max\limits_{v \in V}\left(deg^+\left(v\right)\right)\right).$$
\end{mytheorem}

\begin{proof}

From Lemma~\ref{lem:Descriptors}, there are at most $O(|V|^2)$ descriptors. 
Complexity of all functions which were used in algorithm is the same as in proof of \textbf{Theorem 4} from~\cite{gllParsingTree} except \textbf{Processing} function in which not a single next input token, but the whole set of outgoing edges, should be processed.
Thus, for each descriptor at most $$\max\limits_{v \in V}\left(deg^+\left(v\right)\right)$$ edges  are processed, where $deg^+(v)$ is outdegree of vertex $v$.

Thus, worst-case complexity of proposed algorithm is $$O\left(V^3*\max\limits_{v \in V}\left(deg^+\left(v\right)\right)\right).$$
\end{proof}

%Also we can get averege-case complexity by calculate averege outdegree:
%\begin{align} \label{eq:avg}
%  & O\left(|V|^3*\frac {\sum\limits_{v \in V} deg^+(v)}{|V|}\right) = \nonumber \\
%  & O\left(|V|^2*\sum\limits_{v \in V} deg^+(v)\right) = \nonumber \\
%  & O\left(|V|^2*|E|\right) 
%\end{align}

We can get estimations for linear input from theorem~\ref{thm:complexity}. $\text{For any } v \in V$, $deg^+(v) \leq 1$, thus $\max\limits_{v \in V}(deg^+(v))  = 1 $ and worst-case time complexity $O(|V|^3)$, as expected. 
For LL grammars and linear input complexity should be $O(|V|)$ for the same reason as for original GLL.
 
As discussed in~\cite{modellingGLL}, special data structures, which are required for the basic algorithm, can be not rational for practical implementation, and it is necessary to find balance between performance, software complexity, and hardware resources.
As a result, we can get slightly worse performance than theoretical estimation in practice.

Note that result SPPF contains only paths matched specified query, so result SPPF size is $O(|V'|^3 + |E'|)$ where $M'=(V',E',L')$ is a subgraph of input graph $M$ which contains only matched paths.
Also note that each specific path can be explored by linear SPPF traversal. 

%% file: Example.tex
\subsection{Example}

Let us present a solution for the problem stated in motivating example section (\ref{motivExample}): grammar $G_1$ is a query and we want to find all paths in graph $M$ (presented in picture~\ref{input}) which match this query.
Result SPPF for this input is presented in figure~\ref{SPPF}. Note that presented version does not contains redundant nodes.
Each terminal node corresponds to the edge in the input graph: for each node with label $(v_0, T, v_1)$ there is $e\in E: e=(v_0,T,v_1)$.
We duplicate terminal nodes only for figure simplification.

\begin{figure}[h]
    \begin{center}
        \includegraphics[width=8cm]{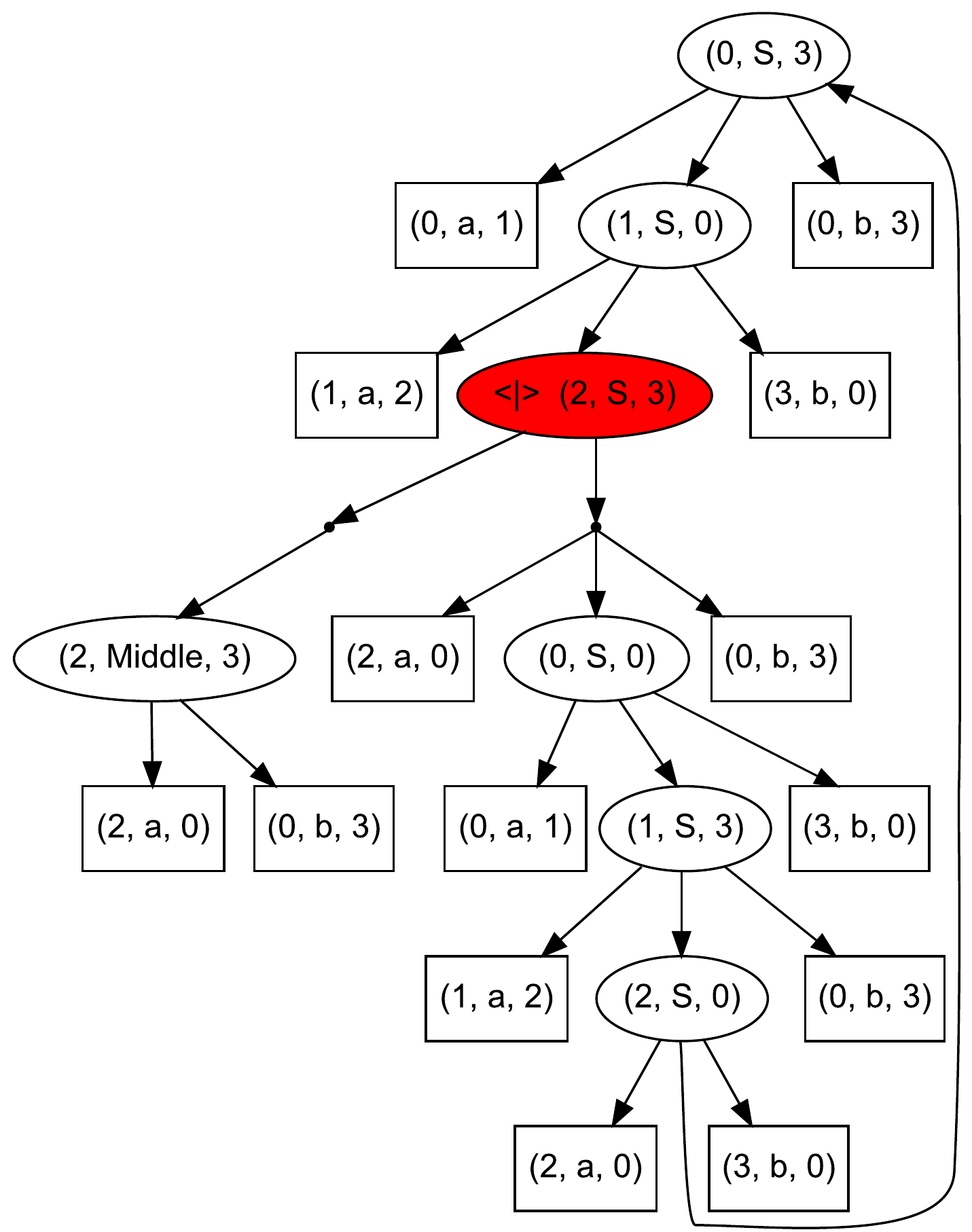}
        \caption{Result SPPF for input graph $M$(fig.~\ref{input}) and query $G_1$(fig.~\ref{grammarG})}
        \label{SPPF}        
    \end{center}
\end{figure}

As an example of derivation structure usage, we can find a middle of any path in example simply by finding correspondent nonterminal \textit{Middle} in SPPF.
So we can find out that there is only one (common) middle for all results, and it is a vertex with $id = 0$. 

Extensions stored in nodes allow us to check whether path from $u$ to $v$ exists and to extract it. 
We need only to traverse SPPF which can be done in polynomial time (in terms of SPPF size) to extract any path . 

Lets find paths $p_i$ such that $S {\xRightarrow[G_1]{}}^{*} \Omega(p_i)$ and $p_i$ starts from the vertex $0$.
To do this, we should find vertices with label $(0, S, \_)$ in SPPF.
(There are two vertices with such labels: $(0, S, 0)$ and $(0, S, 3)$.)
Then let us to extract corresponded paths from SPPF.
There is a cycle in SPPF in our example, so there are \textbf{at least} two different paths: $$p_0=\{(0,a,1);(1,a,2);(2,a,0);(0,b,3);(3,b,0);(0,b,3)\}$$ and 
\begin{align*}
p_1=\{&(0,a,1);(1,a,2);(2,a,0);(0,a,1);(1,a,2);(2,a,0);\\ &(0,b,3);(3,b,0);(0,b,3);(3,b,0);(0,b,3);(3,b,0)\}.
\end{align*}

We demonstrate that SPPF which was constructed by described algorithm can be useful for query result investigation. 
But in some cases explicit representation of matched subgraph is preferable, and required subgraph may be extracted from SPPF trivially by its traversal.

%% file: Evaluation.tex
\section{Evaluation}

In this section we show that performance of implemented algorithm is in good agreement with theoretical estimations, and that the worst-case time and space complexity can be achieved.
We also present the application of our algorithm to the problem of querying RDF ontologies.

All tests were run on a PC with the following characteristics:
\begin{itemize}
\item OS: Microsoft Windows 10 Pro
\item System Type: x64-based PC
\item CPU: Intel(R) Core(TM) i7-4790 CPU @ 3.60GHz, 3601 Mhz, 4 Core(s), 4 Logical Processor(s)
\item RAM: 32 GB
\end{itemize}

\subsection{Ontology querying}

One of classical graph querying problems is a navigation queries for ontologies, and we apply our algorithm to this problem in order to estimate its practical value.
We used dataset from paper~\cite{CFGonRDF}.
Our algorithm is aimed to process graphs, so RDF files were converted to edge-labeled directed graph.
For each triple $(o,p,s)$ form RDF we added two edges: $(o,p,s)$ and $(s,p^{-1},o)$.

We perform two classical \textit{same-generation queries}~\cite{FndDB}.

\textbf{Query 1} is based on the grammar for retrieving concepts on the same layer (presented in figure~\ref{grammarQ1}).
For this query our algorithm demonstrates up to 1000 times better performance and provides identical results as compared to the presented in~\cite{CFGonRDF} for $Q_1$. 

\textbf{Query 2} is based on the grammar for retrieving concepts on the adjacent layers (presented in figure~\ref{grammarQ2}). 
Note that this query differs from the original query $Q_2$ from article~\cite{CFGonRDF} in the following details.
First of all, we count only triples for nonterminal $S$ because only paths derived from it correspond to paths between concepts on adjacent layers.
Algorithm which is presented in~\cite{CFGonRDF} returns triples for all nonterminals.
Moreover, grammar $\mathcal{G}_2$, which is presented in~\cite{CFGonRDF}, describes paths not only between concepts on adjacent layers.
For example, path ``$\text{\textit{subClassOf} \textit{subClassOf}}^{-1}$'' can be derived in $\mathcal{G}_2$, but it is a path between concepts on the same layer, not adjacent.
We changed the grammar to fit a query to a description provided in paper~\cite{CFGonRDF}. 
Thus results of our query is different from results for $Q_2$ which provided in paper~\cite{CFGonRDF}.

Results of both queries are presented in table~\ref{tbl1}, where \#triples is a number of $(o,p,s)$ triples in RDF file, and \#results is a number of triples of form $(S,v_1,v_2)$.
In our approach result triples can be founded by filtering out all SPPF nonterminal nodes labeled by $(v_1,S,v_2)$.

\begin{figure}[ht]
   \begin{center}
   \[
\begin{array}{rl}
   0: & S \rightarrow \text{\textit{subClassOf}}^{-1} \ S \ \text{\textit{subClassOf}} \\ 
   1: & S \rightarrow \text{\textit{type}}^{-1} \ S \ \text{\textit{type}} \\ 
   2: & S \rightarrow \text{\textit{subClassOf}}^{-1} \ \text{\textit{subClassOf}} \\ 
   3: & S \rightarrow \text{\textit{type}}^{-1} \ \text{\textit{type}} \\ 
\end{array}
\]
   \caption{Grammar for query 1}
   \label{grammarQ1}        
   \end{center}
\end{figure}

\begin{figure}[ht]
   \begin{center}
   \[
\begin{array}{rl}
   0: & S \rightarrow B \ \text{\textit{subClassOf}} \\ 
   1: & B \rightarrow \text{\textit{subClassOf}}^{-1} \ B \ \text{\textit{subClassOf}} \\
   2: & B \rightarrow \text{\textit{subClassOf}}^{-1} \ \text{\textit{subClassOf}} \\ 
\end{array}
\]
   \caption{Grammar for query 2}
   \label{grammarQ2}        
   \end{center}
\end{figure}

\begin{table*}
\centering
\caption{Evaluation results for Query 1 and Query 2}
\label{tbl1}

\begin{tabular}{ | c | c | c | c | c | c |}
\hline
Ontology & \#triples & \multicolumn{2}{|c|}{Query 1} & \multicolumn{2}{|c|}{Query 2} \\
\cline{3-6}
& & time(ms) & \#results & time(ms) & \#results \\
\hline 
\hline
skos        & 252 & 10 & 810 & 1 & 1 \\
generations & 273 & 19 & 2164 & 1 & 0 \\
travel      & 277 & 24 & 2499 & 1 & 63 \\
univ-bench  & 293 & 25 & 2540 & 11 & 81 \\
foaf        & 631 & 39 & 4118 & 2 & 10 \\
people-pets & 640 & 89 & 9472 & 3 & 37 \\
funding     & 1086 & 212 & 17634 & 23 & 1158 \\
atom-primitive & 425 & 255 & 15454 & 66 & 122 \\
biomedical-measure-primitive & 459 & 261 & 15156 & 45 & 2871 \\
pizza       & 1980 & 697 & 56195 & 29 & 1262 \\
wine        & 1839 & 819 & 66572 & 8 & 133 \\
\hline
\end{tabular}

\end{table*}

As a result, we conclude that our algorithm is fast enough to be applicable to some real-world problems.

\subsection{Worst-case Complexity} 

We use two grammars for balanced brackets --- ambiguous grammar $G_0$(fig.~\ref{grammarG0}) and unambiguous grammar $G_2$(fig.~\ref{grammarG2}) --- in order to investigate performance and grammar ambiguity correlation.

\begin{figure}[ht]
   \begin{center}
   \[
\begin{array}{rl}
   0: & S \rightarrow a \ S \ b \ S \\ 
   1: & S \rightarrow \varepsilon
\end{array}
\]
   \caption{Unambiguous grammar $G_2$ for balanced brackets}
   \label{grammarG2}        
   \end{center}
\end{figure}

As input we use complete graphs in which for each terminal symbol there is an edge labeled with it between every two vertices.
Note that we use only terminal symbols for edges labels.  
The task we solve in our experiments is to find all paths from all vertices to all vertices satisfied specified query.
Such designed input looks hard for querying in terms of required resources because there is a correct path between any two vertices and result set is infinite.

For complete graph $M=(V,E,L)$ $$\max\limits_{v \in V}\left(deg^+\left(v\right)\right) = (|V| - 1)*|\Sigma|$$, where $\Sigma$ is terminals of input grammar, hence we should get time complexity $O(|V|^4)$ and space complexity $O(|V|^3)$.

Performance measurement results are presented in figure~\ref{pic:Perf}. 
For time measurement results we have that all two curves can be fit with polynomial function of degree 4 to a high level of confidence with $R^2$. 

%g(x) = m*x**3 + n*x**2 + o*x + p
%fit g(x) 'perf/2' using 1:4 via n,m,o,p

\begin{figure}[ht]
\centering
%\begin{gnuplot}
%set terminal epslatex color size 9cm,8cm
%set yrange [0:]
%set key box top left
%set key width 2
%set key opaque
%set sample 1000
%set xlabel 'Number of vertices in input graph'
%set ylabel 'Time in milliseconds'

%f1(x) = 0.000495989*x**4 + 0.001252184*x**3 + 0.068491746*x**2 - 0.306749160*x
%f2(x) = 0.003368883*x**4 - 0.114919298*x**3 + 3.161793404*x**2 - 22.549491142*x

%plot 'perf/2' using 1:3  pt 6 title '$G_2$',\
%     'perf/2' using 1:4  pt 5 title '$G_0$',\
%     f1(x)  with line lt -1 title '$f_1$',\
%     f2(x)  lc rgb "black" dashtype 2 title '$f_2$'     

% \end{gnuplot}
 \includegraphics[width=8cm]{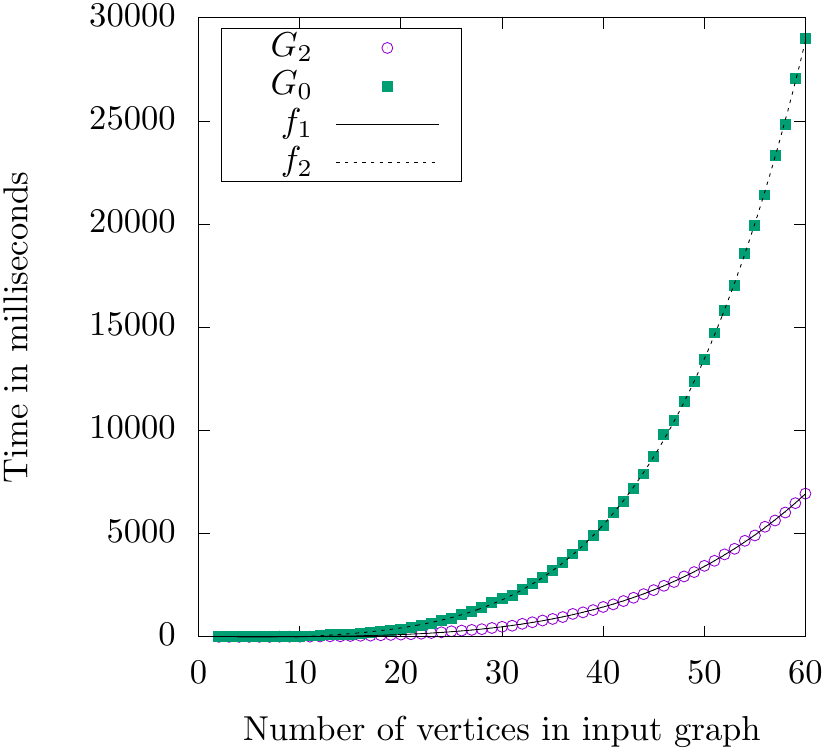}
\caption{Performance on complete graphs for grammar $G_0$ and $G_2$ \\ 
$f_1(x) = 0.000496*x^4 + 0.001252*x^3 + 0.068492*x^2 - 0.306749*x$; $R^2 = 0.99996$ \\
$f_2(x) = 0.003369*x^4 - 0.114919*x^3 + 3.161793*x^2 - 22.54949*x$; $R^2 = 0.99995$}
\label{pic:Perf}
\end{figure}

Also we present SPPF size in terms of nodes for both $G_0$ and $G_2$ grammars (fig.~\ref{pic:SPPFSize}).
As was expected, all two curves are cubic to a high level of confidence with $R^2 = 1$. 

\begin{figure}[ht]
\centering
%\begin{gnuplot}
%set terminal epslatex color size 9cm,8cm
%set key box top left
%set key width 2
%set key opaque
%set sample 1000
%set xlabel 'Number of vertices in input graph'
%set ylabel 'Number of SPPF nodes'

%f1(x) = 3.000047*x**3 + 3.994579*x**2 + 4.191568*x
%f2(x) = 3.000050*x**3 + 2.994338*x**2 + 4.196472*x

%plot 'perf/2' using 1:6 pt 6 title '$G_2$',\
%     'perf/2' using 1:7 pt 5 title '$G_0$',\
%     f1(x)  with line lt -1 title '$f_1$',\
%     f2(x)  lc rgb "black" dashtype 2 title '$f_2$'     

% \end{gnuplot}
 \includegraphics[width=8cm]{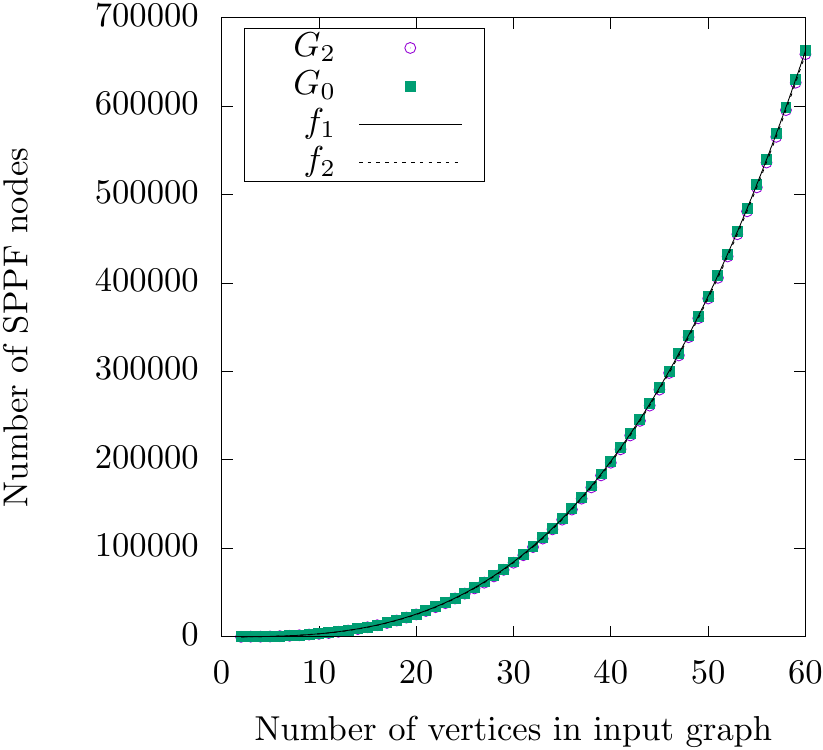}
\caption{SPPF size on complete graph for grammar $G_0$ and $G_2$ a complete graphs \\
$f_1(x) = 3.000047*x^3 + 3.994579*x^2 + 4.191568*x$; $R^2 = 1$\\
$f_2(x) = 3.000050*x^3 + 2.994338*x^2 + 4.196472*x$; $R^2 = 1$}
\label{pic:SPPFSize}
\end{figure}

%\begin{figure}[h]
%\centering
%\begin{gnuplot}
%set terminal epslatex color size 9cm,8cm
%set key box top left
%set logscale y
%set key width 2
%set key opaque
%set sample 1000
%set xlabel '$x$-label'
%set ylabel '$y$-label
%plot 'perf/3' using 1:2 with lines ls 2 ti '$Unamb$',\
%     'perf/3' using 1:3 with lines ls 3 ti '$amb$'
%
% \end{gnuplot}
%\caption{Performance on C graph for grmmars $G_0$ and $G_2$}
%\label{pic:DoubleCyclesPerf}
%\end{figure}

%To summarise we can say that performance for unambiguos grammars is better then for ambiguos. 

%Full graphs for balanced brackets.

%Full graph for highly unambiguos greammar $G_3$ (figure~\ref{grammarG3}).

%\begin{figure}[h]
%   \begin{center}
%\begin{verbatim}
%   0: s = s s s 
%   1: s = s s
%   2: s = A
%\end{verbatim}
%   \caption{Highly ambiguos grammar $G_3$}
%   \label{grammarG3}        
%   \end{center}
%\end{figure}

%% file: Conclusion.tex
\section{Conclusion and Future Work}

We propose GLL-based algorithm for context-free path querying which constructs finite structural representation of all paths satisfying given constraint.
Provided data structure can be useful for result investigation and processing, and for query debugging.
Presented algorithm has been implemented in F\# programming language~\cite{FSharp} and is available on GitHub:\url{https://github.com/YaccConstructor/YaccConstructor}.

In order to estimate practical value of proposed algorithm, we should perform evaluation on a real dataset and real queries.
One possible application of our algorithm is metagenomical assembly querying, and we are currently working on this topic.

We are also working on performance improvement by implementation of recently proposed modifications in original GLL algorithm~\cite{FGLL,FastPracticalGLL}.
One direction of our research is generalization of grammar factorization proposed in~\cite{FGLL} which may be useful for the processing of regular queries which are common in real world application.

%We are working on utilization of GPGPU and multicore CPU power for graph parsing problem with Valiant~\cite{valiantParsingWithMatrixMultiplication} algorithm modification proposed by Alexander Okhotin~\cite{okhotin2014parsingWithMatrixMultiplication}.
%One of possible benefit is ability to process more expressive queries because modification proposed by Alexander Okhotin extended to support boolean grammars.

%% file: appendix.tex
\appendix

\section{GLL pseudocode}\label{GLLCode}

Main functions of GLL parsing algorithms: 
\begin{itemize}
\item Algorithm~\ref{stackF}---stack and descriptors manipulation functions;
\item Algorithm~\ref{sppfF}---SPPF construction functions.
\end{itemize}

Used notation:
\begin{itemize}
\item $(L, s, j, a)$---descriptor, where $L$ is a grammar slot, $s$ is a stack node, $j$ is a position in the input string, and $a$ is a node of derivation tree;
\item $R$---working set which contains descriptors to process;
\item $U$---all descriptors was created;
\item $P$---popped nodes.
\end{itemize}

\begin{algorithm}
\begin{algorithmic}[1]
\caption{Stack and descriptors manipulation}
\label{stackF}
\Function{add}{$L,v,i,a$}
  \If{$(L,v,i,a) \notin U$}  
      \State{$U.add(L,v,i,a)$}
      \State{$R.add(L,v,i,a)$}
  \EndIf
\EndFunction

\Function{pop}{$v,i,z$}
  \If{$v \neq v_0$}  
      \State{$P.add(v,z)$}
      \ForAll{$(a,u) \in v.outEdges$}
        \State{$y \gets$ \Call{getNodeP}{$v.L, a, z$}}
        \State{\Call{add}{$v.L,u,i,y$}}
      \EndFor
  \EndIf
\EndFunction

\Function{create}{$L,v,i,a$}
  \If{$(L,i) \notin GSS.nodes$}  
      \State{$GSS.nodes.add (L,i)$}
  \EndIf
  \State{$u \gets$ $GSS.nodes.get(L, i)$}
  \If{$(u,a,v) \notin GSS.edges$}  
      \State{$GSS.edges.add(u,a,v)$}
      \ForAll{$(u,z) \in P$}
         \State{$y \gets$ \Call{getNodeP}{$L, a, z$}}
         \State{$(\_,\_, k) \gets z.lbl$}
         \State{\Call{add}{$L,v,k,y$}}
      \EndFor
  \EndIf
  \Return{$u$}
\EndFunction

\end{algorithmic}
\end{algorithm}

\begin{algorithm}
\begin{algorithmic}[1]
\caption{SPPF construction}
\label{sppfF}

\Function{getNodeT}{$x,i$}
  \If{$x = \varepsilon$}  
      \State{$h \gets i$}
  \Else
      \State{$h \gets i + 1$}
  \EndIf
  \If{$(x,i,h) \notin SPPF.nodes$}  
      \State{$SPPF.nodes.add(x,i,h)$}
  \EndIf
  \State{\Return{$SPPF.nodes.get(x,i,h)$}}
\EndFunction

\Function{getNodeP}{$(X \rightarrow \omega_1 \cdot \omega_2), a ,z$}
  \If{$\omega_1$ is terminal or non-nullable nonterminal and $\omega_2 \neq \varepsilon$}  
      \State{\Return{$z$}}
  \Else
      \If{$\omega_2 = \varepsilon$}  
          \State{$t \gets X$}
       \Else
          \State{$h \gets (X \rightarrow \omega_1 \cdot \omega_2)$}
       \EndIf
       \State{$(q,k,i) \gets z.lbl$}
      \If{$a \neq dummy$}  
         \State{$(s,j,k) \gets a.lbl$}
         \State{$y \gets findOrCreate \ SPPF.nodes \ (n.lbl = (t,i,j))$}
         \If{$y$ does not have a child labeled $(X \rightarrow \omega_1 \cdot \omega_2)$}
             \State{$y' \gets newPackedNode(a,z)$}
             \State{$y.chld.add \ y'$}
             \State{\Return{$y$}}
         \Else
             \State{$y \gets findOrCreate \ SPPF.nodes \ (n.lbl = (t,k,i))$}
             \If{$y$ does not have a child labeled $(X \rightarrow \omega_1 \cdot \omega_2)$}
                \State{$y' \gets newPackedNode(z)$}
                \State{$y.chld.add \ y'$}
                \State{\Return{$y$}}
             \EndIf
         \EndIf
      \EndIf
  \EndIf
  \State{\Return{$SPPF.nodes.get(x,i,h)$}}
\EndFunction

\end{algorithmic}
\end{algorithm}

%% file: ContextFreeConstrainedPathFindingInGraph.bbl
\begin{thebibliography}{10}

\bibitem{FndDB}
S.~Abiteboul, R.~Hull, and V.~Vianu.
\newblock Foundations of databases, 1995.

\bibitem{FastPracticalGLL}
A.~Afroozeh and A.~Izmaylova.
\newblock Faster, practical gll parsing.
\newblock In {\em International Conference on Compiler Construction}, pages
  89--108. Springer, 2015.

\bibitem{Alvor1}
A.~Annamaa, A.~Breslav, J.~Kabanov, and V.~Vene.
\newblock An interactive tool for analyzing embedded sql queries.
\newblock In {\em Asian Symposium on Programming Languages and Systems}, pages
  131--138. Springer, 2010.

\bibitem{QueryGraphWithData}
P.~Barcel{\'o}, G.~Fontaine, and A.~W. Lin.
\newblock Expressive path queries on graphs with data.
\newblock In {\em International Conference on Logic for Programming Artificial
  Intelligence and Reasoning}, pages 71--85. Springer, 2013.

\bibitem{FLCpathProblem}
C.~Barrett, R.~Jacob, and M.~Marathe.
\newblock Formal-language-constrained path problems.
\newblock {\em SIAM Journal on Computing}, 30(3):809--837, 2000.

\bibitem{LLnonLL}
J.~C. Beatty.
\newblock Two iteration theorems for the ll (k) languages.
\newblock {\em Theoretical Computer Science}, 12(2):193--228, 1980.

\bibitem{Alvor2}
A.~Breslav, A.~Annamaa, and V.~Vene.
\newblock Using abstract lexical analysis and parsing to detect errors in
  string-embedded dsl statements.
\newblock In {\em Proceedings of the 22nd Nordic Workshop on Programming
  Theory}, pages 20--22, 2010.

\bibitem{TableGLL}
S.~V. Grigorev and A.~K. Ragozina.
\newblock Generalized table-based ll-parsing.
\newblock {\em Sistemy i Sredstva Informatiki [Systems and Means of
  Informatics]}, 25(1):89--107, 2015.

\bibitem{ConjCFPathQuery}
J.~Hellings.
\newblock Conjunctive context-free path queries.
\newblock 2014.

\bibitem{hofman2015separabilityForRegQueryDebugging}
P.~Hofman and W.~Martens.
\newblock Separability by short subsequences and subwords.
\newblock In {\em LIPIcs-Leibniz International Proceedings in Informatics},
  volume~31. Schloss Dagstuhl-Leibniz-Zentrum fuer Informatik, 2015.

\bibitem{modellingGLL}
A.~Johnstone and E.~Scott.
\newblock Modelling gll parser implementations.
\newblock In {\em International Conference on Software Language Engineering},
  pages 42--61. Springer Berlin Heidelberg, 2010.

\bibitem{SPPF}
J.~G. Rekers.
\newblock {\em Parser generation for interactive environments}.
\newblock PhD thesis, Citeseer, 1992.

\bibitem{RegularDBQuery}
J.~L. Reutter, M.~Romero, and M.~Y. Vardi.
\newblock Regular queries on graph databases.
\newblock {\em Theory of Computing Systems}, pages 1--53, 2015.

\bibitem{scott2010gll}
E.~Scott and A.~Johnstone.
\newblock Gll parsing.
\newblock {\em Electronic Notes in Theoretical Computer Science},
  253(7):177--189, 2010.

\bibitem{gllParsingTree}
E.~Scott and A.~Johnstone.
\newblock Gll parse-tree generation.
\newblock {\em Science of Computer Programming}, 78(10):1828--1844, 2013.

\bibitem{FGLL}
E.~Scott and A.~Johnstone.
\newblock Structuring the gll parsing algorithm for performance.
\newblock {\em Science of Computer Programming}, 125:1--22, 2016.

\bibitem{brnglr}
E.~Scott, A.~Johnstone, and R.~Economopoulos.
\newblock Brnglr: a cubic tomita-style glr parsing algorithm.
\newblock {\em Acta informatica}, 44(6):427--461, 2007.

\bibitem{GraphQueryWithEarley}
P.~Sevon and L.~Eronen.
\newblock Subgraph queries by context-free grammars.
\newblock {\em Journal of Integrative Bioinformatics}, 5(2):100, 2008.

\bibitem{FSharp}
D.~Syme, A.~Granicz, and A.~Cisternino.
\newblock {\em Expert F\# 3.0}.
\newblock Springer, 2012.

\bibitem{Tomita}
M.~Tomita.
\newblock An efficient context-free parsing algorithm for natural languages.
\newblock In {\em Proceedings of the 9th International Joint Conference on
  Artificial Intelligence - Volume 2}, IJCAI'85, pages 756--764, San Francisco,
  CA, USA, 1985. Morgan Kaufmann Publishers Inc.

\bibitem{relaxedRNGLR}
E.~Verbitskaia, S.~Grigorev, and D.~Avdyukhin.
\newblock Relaxed parsing of regular approximations of string-embedded
  languages.
\newblock In {\em International Andrei Ershov Memorial Conference on
  Perspectives of System Informatics}, pages 291--302. Springer International
  Publishing, 2015.

\bibitem{CFGonRDF}
X.~Zhang, Z.~Feng, X.~Wang, G.~Rao, and W.~Wu.
\newblock Context-free path queries on rdf graphs.
\newblock {\em arXiv preprint arXiv:1506.00743}, 2015.

\end{thebibliography}
